\documentclass[conference]{IEEEtran}
\IEEEoverridecommandlockouts

\usepackage{geometry}
\usepackage{cite}
\usepackage{amsmath,amssymb,amsfonts}
\usepackage{algorithmic}
\usepackage{graphicx}
\usepackage{textcomp}
\usepackage{xcolor}
\usepackage{amsthm}
\newtheorem{theorem}{Theorem}

\newtheorem{lemma}{Lemma}
\newtheorem{corollary}{Corollary}
\newtheorem{definition}{Definition}
\newtheorem{remark}{Remark}

\def\BibTeX{{\rm B\kern-.05em{\sc i\kern-.025em b}\kern-.08em
		T\kern-.1667em\lower.7ex\hbox{E}\kern-.125emX}}
\begin{document}
	
	\title{About Optimal Prefix Codes over Countably Infinite Alphabets: Probabilistic Intervals for the Codeword Lengths Assignment\\
    \thanks{This work was supported by the National Natural Science Foundation of China under Grant 62401599.(Correspondence author: Wei Yan)}
    }
	
	\author{\IEEEauthorblockN{Hongyang Liu\IEEEauthorrefmark{1}\IEEEauthorrefmark{2}and Wei Yan\IEEEauthorrefmark{1}\IEEEauthorrefmark{2}}
		\IEEEauthorblockA{\IEEEauthorrefmark{1}College of Electronic Engineering, National University of Defense Technology, Hefei, China\\
			\IEEEauthorrefmark{2}Anhui Province Key Laboratory of Cyberspace Security Situation Awareness and Evaluation, Hefei, China\\
			liuhongyangl@nudt.edu.cn, yan.wei2023@nudt.edu.cn}
	}
	
	\maketitle
	\begin{abstract}
		For the discrete memoryless sources with a countably infinite alphabet, we prove that for any positive integer $k$, there exists a corresponding probability interval such that if the largest symbol probability $p_{1}$ falls in this interval, the optimal code length for the symbol equals $k$. Furthermore, for infinite sources, we provide a criterion to determine probability distributions whose optimal code length assignment follows the pattern $l^{best}_{i}=i$, for $i\ge 1$. Compared with the existing conclusion for anti-uniform sources, the proposed criterion requires less information for verification.
	\end{abstract}
	
	
\section{Introduction and Motivation}
Considering the discrete memoryless source $\mathcal{S}=(\mathcal{X},\mathcal{P})$, where the alphabet $\mathcal{X}=\{1,2,\ldots\}$ is countable and probability distribution $\mathcal{P}=(p_{1},p_{2},\ldots)$ satisfies $\sum_{i=1}^{\infty}p(i)=1$ and
\begin{equation}\label{eq1}
p(i)\geq p(i+1)\geq 0, \quad\quad i\in\mathcal{X}.
\end{equation}
If the alphabet is finite, the famous Huffman coding algorithm provides an optimal prefix code\cite{Huffman:1952}. However, the bottom-up construction of the Huffman algorithm prevents its direct application to countably infinite alphabets. In particular, for probability distributions with certain properties, optimal prefix codes for infinite sources can be inferred by generalizing the Huffman algorithm's properties.

Perhaps the earliest code of this type is the optimal code for the sources with geometric distribution, developed by Golomb\cite{Golomb:1966} and Gallager \textit{et al.}\cite{Gallager:1975}. Subsequently, in later work by Humblet\cite{Humblet:1978} on the optimal code for Poisson distribution, it was shown that for any distribution whose tail decays faster than $(\frac{\sqrt{5}-1}{2})^i\approx0.618^{i}$, the coding problem can be reduced to that of Huffman coding on a finite source. Later, Abrahams\cite{Abrahams:1994} extended the results and presented that
for some positive integers $M$,
there exists an integer $m$ such that the set of inequalities
\begin{equation}\label{eq2}
	 \sum_{i=1}^{\infty}p_{k+iM+1}\leq  p_{k} \leq \sum_{i=1}^{\infty}p_{k+iM-1},\quad k=m,m+1,\ldots
\end{equation}
holds, then the optimal code can be obtained via Huffman codes of suitably reduced sources,
where the source is the finite alphabet constructed by grouping and merging the tail probabilities of the original countably infinite symbol set. Further, Akiko Kato\cite{Kato:1996} proposed new sufficient conditions under which, for a broader class of probability distributions, an optimal $D$-ary prefix code can be constructed via the Huffman algorithm.

In a remarkable paper, Linder \textit{et al.}\cite{Linder:1997} introduced the notion of \textit{convergence} for code sequences,
and proved that for the source $\mathcal{S}$ with a countably infinite alphabet $\mathcal{X}$, Huffman codes for its truncated version $\mathcal{S}_{n}$ with the alphabet $\mathcal{X}_{n}=\{1,2,...,n\}$ contain a subsequence converging to the optimal prefix code for $\mathcal{S}$. Moreover, this optimal code satisfies the equality condition of the Kraft inequality. This result establishes a pathway for extending Huffman-based methods to study optimal codes for infinite sources\cite{Linder:1997,Chow:1998}.
	
In fact, Golin\cite{Golin:2003} took a different approach by mapping the code tree of an infinite source onto a path in an infinite weighted graph from the algorithmic structure perspective. He demonstrated that the optimal infinite tree exhibits a repetitive structure described as a finite head followed by a finite periodic segment, thereby transforming the optimal coding problem into a search for the shortest cyclic path under certain constraints. This observation leads us to examine the problem from the perspective of the tree structure.
	
Another class of discrete memoryless sources was introduced by Esmaeili \textit{et al}.\cite{Esmaeili:2005,Esmaeili:2007},
a source $\mathcal{S}_n$ with a finite alphabet $\mathcal{X}_{n}=\{1,2,\ldots,n\}$ is called \textit{anti-uniform} if $\mathcal{S}_n$ has a Huffman code, its codeword length satisfies $l_{1}=1,l_{2}=2,\ldots,l_{n-1}=l_{n}=n-1$.
This notion was then extended to countably infinite source, whose optimal code length is $\{l_{i}=i\}_{i=1}^{\infty}$.
The sufficient condition\footnote{Reference~\cite{Esmaeili:2007} claims that~\eqref{eq3} is a necessary and sufficient condition for an infinite source to be anti-uniform, but no relevant proof is provided.} for an infinite source to be anti-uniform has been established as follows\cite{Humblet:1978}.
\begin{equation}\label{eq3}
		p_{i+2}+p_{i+3}+\cdots=\displaystyle\sum_{k=i+2}^{\infty}p_{k}\le p_{i}, \quad i\in\mathcal{X}.
\end{equation}

However, the optimal prefix coding for a source with a countably infinite alphabet remains an open problem.
In this paper, we address two aspects of optimal code for the infinite alphabet. First, by analyzing the structure of Huffman trees in the finite alphabet setting, we prove that for an infinite source, the optimal code length of the symbol with the largest probability is uniquely determined whenever the probability falls in certain specified intervals. Second, we introduce a new criterion for identifying anti-uniform sources.
The main contributions of this study are summarized below.
	\begin{enumerate}
	\item  We investigate the optimal codes for countably infinite sources and obtain the probability intervals for which the optimal codeword length corresponding to the maximum probability equals $k$.
	\item We propose a recursive criterion for identifying anti-uniform sources, which requires less information for verification.
	\end{enumerate}

\section{Preliminary}
In this section, we first define the necessary notations for source coding, then review the existence conclusions for optimal prefix codes of infinite sources, and finally discuss an essential property of optimal prefix codes.
\subsection{Notations}
Let $\mathcal{S}\triangleq(\mathcal{X},\mathcal{P})$ denote a discrete memoryless source, where $\mathcal{X}=\{1,2,\ldots\}$ denotes a countable alphabet, and $\mathcal{P}=(p_{1},p_{2},\ldots)$ represents the corresponding probability distribution
satisfying $\sum_{i=1}^{\infty}p(i)=1$ and~\eqref{eq1}.
If the alphabet is countably infinite, the source is referred to as an infinite source.
Let $C_{best}$ denote its optimal prefix code, and the corresponding codeword length assignment is $\mathcal{L}_{best}=\{l^{best}_{i}\}_{i=1}^{\infty}$.

If the alphabet is finite with size $n$, then the source is called a finite source, denoted by $\mathcal{S}_{n}=(\mathcal{X}_{n},\mathcal{P}_{n})$, where $\mathcal{X}_{n}=\{1,2,...,n\}$ and $\mathcal{P}=(p_{1},p_{2},...,p_{n})$. Correspondingly, Huffman code for $\mathcal{S}_{n}$ is denoted by $C_{n}$, and the associated codeword length assignment is $\mathcal{L}_{n}=\{l_{i,n}\}_{i=1}^{n}$. If the discussion of the alphabet size is not involved, the Huffman code lengths are simply denoted as $\mathcal{L}_{n}=\{l_{i}\}_{i=1}^{n}$.

\subsection{Convergence and Properties of the Code Scheme}
\begin{definition}\cite{Linder:1997}: \label{definition1}
    A sequence of codes $\{C_{n}\}_{n=1}^{\infty}$ converges to a code $C$, denoted by
\begin{equation}
	C_{n}\rightarrow C,\,  n\rightarrow\infty,\label{eq5}
\end{equation}
	if and only if for any $i>0$, there exists an integer $N$ such that for any $n>N$, the $i$-th codeword of $C_{n}$ is eventually constant to the $i$-th codeword of $C$.
   \end{definition}
	Based on Definition~\ref{definition1}, the following theorem is established.
	\begin{theorem}\cite{Linder:1997}: \label{theorem1}
		For the infinite source $\mathcal{S}$, its truncated version of size $n$ is proposed as $\mathcal{S}_{n}=(\mathcal{X}_{n},\mathcal{P}_{n})$, where $\mathcal{P}_{n}=(q_{1},q_{2},...,q_{n})$,
		\begin{equation}
			q_{i}\triangleq p_{i}/S_{n},\, S_{n}\triangleq\sum_{j=1}^{n}p_{j}.\label{eq6}
		\end{equation}
		The sequence $\{C_{n}\}_{n=1}^{\infty}$ contains a subsequence that converges to the optimal prefix code $C_{best}$.
	\end{theorem}
	Following the review of convergence results for the code scheme, we note that Huffman coding can generate optimal prefix codes, and the optimal prefix coding problem possesses both the greedy-choice property and the optimal-substructure property\cite{Cormen:2022}. The optimal-substructure property is introduced below.
	\begin{theorem}\cite{Cormen:2022}: \label{theorem2}
		Consider a finite source $\mathcal{S}_{n}=(\mathcal{X}_{n},\mathcal{P}_{n})$. Let $x,y\in \mathcal{X}_{n}$ be two symbols with the smallest probabilities. Define a modified source $\mathcal{S}'=(\mathcal{X}',\mathcal{P}')$, where
		\begin{equation}
			\mathcal{X}'=\mathcal{X}_{n}-\{x,y\}\cup \{z\}\label{eq7}
		\end{equation}
		and $\mathcal{P}'$ is identical to $\mathcal{P}_{n}$ except that $p'_{z}=p_{x}+p_{y}$. Suppose $C'$ is an optimal prefix code for $\mathcal{S}'$. Then, replacing the leaf node for $z$ in the code tree of $C'$ with an internal node that has $x,y$ as children can obtain an optimal prefix code $C_{n}$ for the source $\mathcal{S}_{n}$.
	\end{theorem}
	Furthermore, a specific class of Huffman codes with distinctive characteristics has received particular research attention, known as anti-uniform sources\cite{Esmaeili:2005,Esmaeili:2007}. The following condition is provided.
\begin{theorem}\cite{Esmaeili:2005}: \label{theorem3}
A finite source admits an anti-uniform Huffman code if and only if the following condition
	\begin{equation}
			p_{i+2}+p_{i+3}+\cdots+p_{n}\le p_{i},\quad 1\le i\le n-3.\label{eq8}
	\end{equation}
holds.
\end{theorem}
\section{Properties of the Standardized Huffman Algorithm}
In this section, we first define the standardized Huffman algorithm with a unique merging rule. Then, by studying the structural properties of Huffman code trees, we present three lemmas regarding the codeword length of the symbol with the maximum probability.

The bottom-up merging operation of the Huffman algorithm, which repeatedly combines the two nodes with the smallest probabilities, can yield different results for the same source. To ensure the uniqueness of the coding result and thereby facilitate subsequent analysis and proofs, we define a standardized merging rule.

Given a finite source $\mathcal{S}_{n}=(\mathcal{X}_{n},\mathcal{P}_{n})$, let
\begin{equation}\label{eq9}
	\mathcal{P}_{n}^{(m)}\triangleq(p_{1}^{(m)},p_{2}^{(m)},\ldots,p_{n-m}^{(m)})
\end{equation}
denote the probability distribution after $m$ merge operations,
where $p_{1}^{(m)}\ge p_{2}^{(m)}\ge\cdots\ge p_{n-m}^{(m)}$ and $m=0,1,\ldots,n-1$.
In particular, $\mathcal{P}_{n}^{(0)}=\mathcal{P}_{n}$.
The standardized Huffman merging rules are defined as follows.

\begin{definition}\label{definition3}
For any $m\ge0$, let $p_{0}^{(m)}\triangleq\alpha>1$.
During the merging step applied to $\mathcal{P}_{n}^{(m)}$, the two nodes $p_{n-m}^{(m)},p_{n-m-1}^{(m)}$ are selected.
There exists an integer $k$ that satisfies the following inequality.
\begin{equation}
			p_{k-1}^{(m)}>p_{n-m-1}^{(m)}+p_{n-m}^{(m)}\ge p_{k}^{(m)}.\label{eq10}
\end{equation}
The probability distribution $$\mathcal{P}_{n}^{(m+1)}=(p_{1}^{(m+1)},p_{2}^{(m+1)},\ldots,p_{n-m-1}^{(m+1)})$$
obtained after this merger is given by the following equation.
	\begin{equation*}
			p_{i}^{(m+1)}=\left\{
			\begin{aligned}
				&p_{i}^{(m)},&& 0\le i \le k-1\\
				&p_{n-m}^{(m)}+p_{n-m-1}^{(m)},&& i=k \\
				&p_{i-1}^{(m)},&& k+1\le i\le n-m-1
			\end{aligned}.
			\right.
	\end{equation*}
	\end{definition}
In practice, after standardizing the merging procedure of the Huffman algorithm, we pay particular attention to a specific moment in the merging process, defined as follows.
\begin{definition}\label{definition4}
For a finite source $\mathcal{S}_{n}=(\mathcal{X}_{n},\mathcal{P}_{n})$, where $\mathcal{P}_{n}$ satisfies
\begin{equation}\label{eq12}
			 p_{1}\ge p_{2}\ge\cdots\ge p_{n}.
\end{equation}
	If $p_{1}<\frac{1}{2}$,\footnote{If $p_{1}\ge \frac{1}{2}$, $l_{1}$ is necessarily 1. This is regarded as a trivial case and thus excluded from further consideration.} then let $\delta$ be the smallest integer such that $\mathcal{P}_{n}^{(\delta-1)},\mathcal{P}_{n}^{(\delta)}$ satisfy
		\begin{equation}
			p_{n-\delta+1}^{(\delta-1)}+p_{n-\delta}^{(\delta-1)}<p_{1},\label{eq13}
		\end{equation}
		\begin{equation}
			p_{n-\delta}^{(\delta)}+p_{n-\delta-1}^{(\delta)}\ge p_{1}.\label{eq14}
		\end{equation}
		This moment is then defined as the $\delta$-occasion. In particular, when $p_{n-1}+p_{n}\ge p_{1}$, we define $\delta=0$.
	\end{definition}
By investigating the $\delta$-occasion, we gain insight into some favorable properties of the standardized Huffman code tree, which allows us to prove key lemmas regarding the value of $l_{1}$. This is shown in the following lemmas.
\begin{lemma}\label{lemma1}
Consider a finite source $\mathcal{S}_{n}=(\mathcal{X}_{n},\mathcal{P}_{n})$, where $\mathcal{P}_{n}$ satisfies \eqref{eq12} and $p_{1}<\frac{1}{2}$. On the $\delta$-occasion of the standardized Huffman coding procedure.
	\begin{enumerate}
		\item If $p_{1}<b$, then $\delta<n-\frac{1}{b}$; \vspace{2pt}
		\item if $p_{1}>a$, then $\delta>n-\frac{2}{a}+1$.
	\end{enumerate}
	\end{lemma}
\begin{proof}The lemma is proved by contradiction.
\begin{enumerate}	
	\item If $p_{1}<b$ and $\delta\ge n-\frac{1}{b}$, then for the probability distribution $\mathcal{P}_{n}^{(\delta)}$, the inequality
			\begin{equation*}
              \begin{aligned}
				\sum_{j=1}^{n-\delta}p_{j}^{(\delta)}& \le (n-\delta)p_{1}^{(\delta)}  \\
                                        & =(n-\delta)p_{1}   \\
                                         & <\frac{1}{b}\times b=1  \\
              \end{aligned}
			\end{equation*}
		holds. This leads to a contradiction.
	\item If $p_{1}>a$ and $\delta\le n-\frac{2}{a}+1$, then for the  probability distribution $\mathcal{P}_{n}^{(\delta)}$, from \eqref{eq14}, it follows that for any $t\le n-\delta-1$, the inequality
		\begin{equation*}
			2p_{t}^{(\delta)}\ge p_{n-\delta-1}^{(\delta)}+p_{n-\delta}^{(\delta)}\ge p_{1}
		\end{equation*}
holds. Consequently, we obtain the following inequality.
	\begin{equation*}
    \begin{aligned}
				\sum_{j=1}^{n-\delta}p_{j}^{(\delta)}& \ge p_{1}+(n-\delta-3)\frac{p_{1}}{2}+p_{1}  \\
                                      & = \frac{n-\delta+1}{2}p_{1}   \\
                                      & >\frac{1}{a}\times a=1,  \\
   \end{aligned}
	\end{equation*}
which yields a contradiction. This completes the proof.
	\end{enumerate}
	\end{proof}
	Building on Lemma~\ref{lemma1}, $l_{1}$ can be determined on the $\delta$-occasion of the standardized Huffman coding procedure.
	\begin{figure}[!t]
		\centering
		\includegraphics[width=0.4\textwidth]{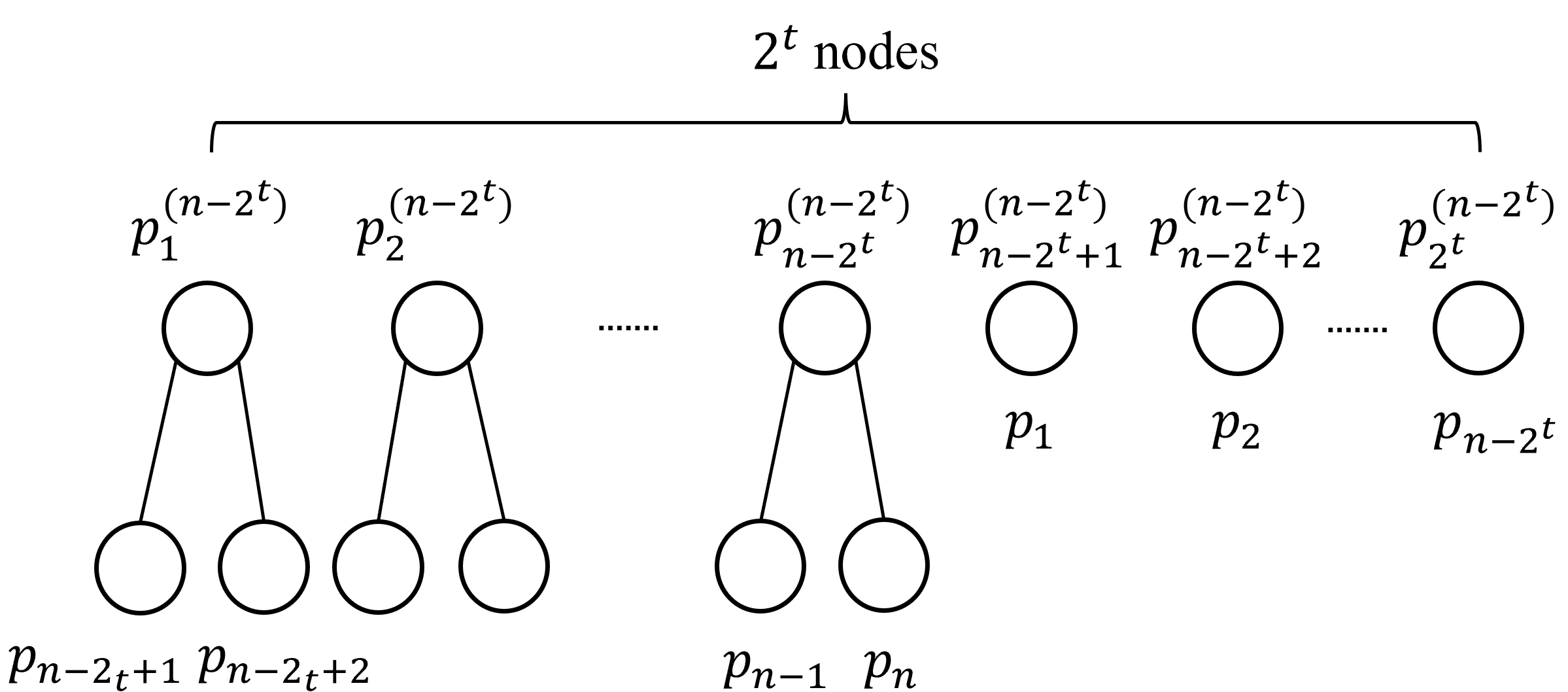}
		\caption{The change of probability distribution during the first $n-2^{t}$ merges.}
		\label{phase1}
	\end{figure}
	\begin{figure}[!t]
		\centering
		\includegraphics[width=0.4\textwidth]{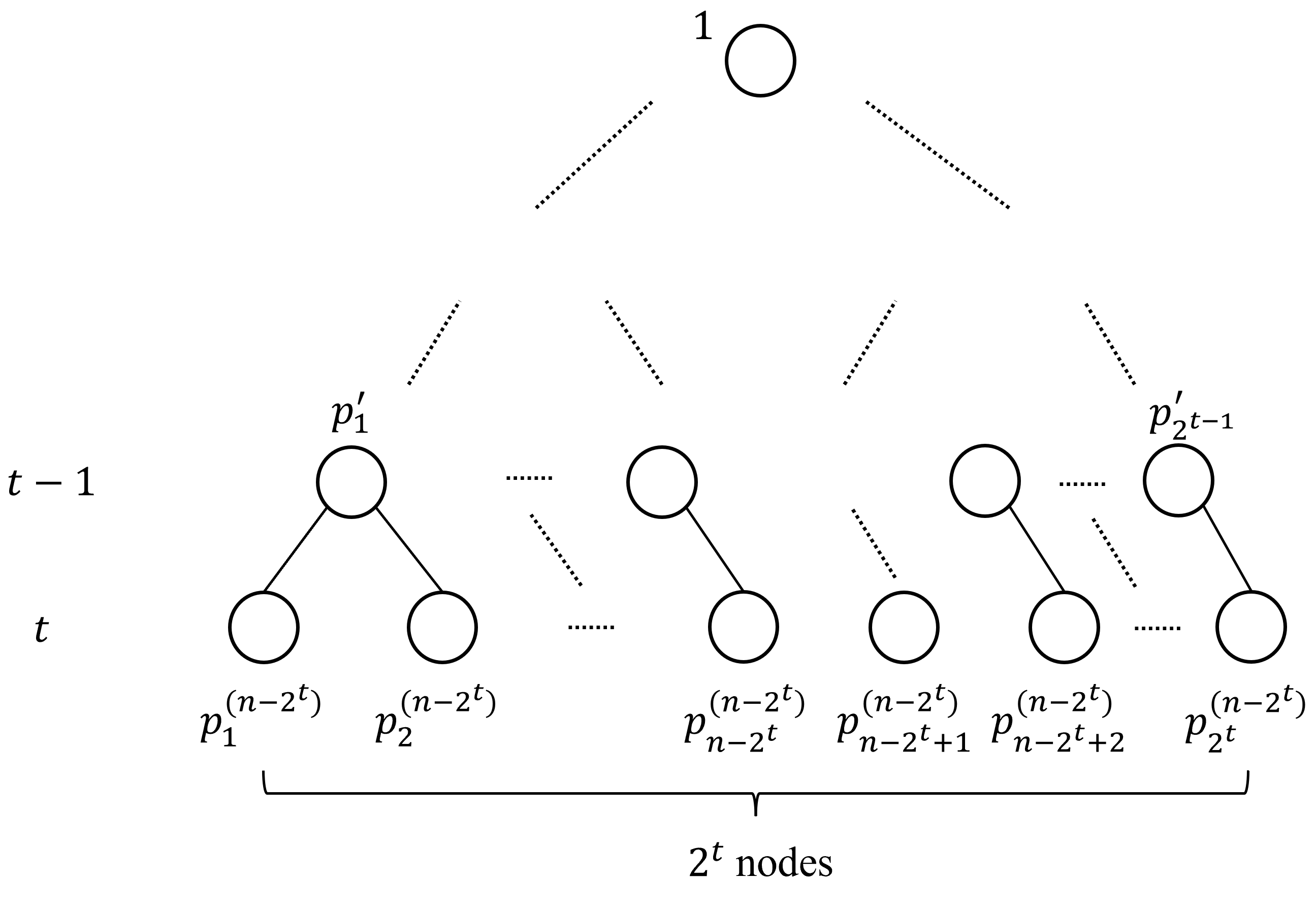}
		\caption{The change of probability distribution during the last $2^{t}$ merges.}
		\label{phase2}
	\end{figure}
	\begin{lemma}\label{lemma2}
		The probability distribution $\mathcal{P}_{n}$ satisfies~\eqref{eq12}. If $p_{n-1}+p_{n}\ge p_{1}$, then $l_{1}=\lfloor \log_{2}n\rfloor$.
	\end{lemma}
\begin{proof}
Let $t=\lfloor\log_{2}n\rfloor$, then $2^{t}\le n\le 2^{t+1}-1$. Consider the merging process, which can be divided into two phases.
	\begin{enumerate}
\item Consider the first $n-2^{t}$ merges. The following expression shows the change of probability distribution during the first $n-2^{t}$ merges.
		\begin{small}
			\begin{equation}
				\begin{aligned}
					&\mathcal{P}_{n}^{(0)}=(p^{(0)}_{1},p^{(0)}_{2},\ldots,p^{(0)}_{n})=(p_{1},p_{2},\ldots,p_{n}),\\
					&\mathcal{P}_{n}^{(1)}=\big(p^{(1)}_{1}(=p_{n-1}+p_{n}),p^{(1)}_{2}(=p_{1}),\ldots,p^{(1)}_{n-1}\big),\\
					&~~~~~~~~~~~~~~~~~~~~~~~~~~~~~~~~~~~~\vdots \\
					&\mathcal{P}_{n}^{(i)}=(p_{1}^{(i)}(=p_{n-i}^{(i-1)}+p_{n-i+1}^{(i-1)}),\ldots,p_{i+1}^{(i)}(=p_{1}),\ldots,p_{n-i}^{(i)}),\\
					&~~~~~~~~~~~~~~~~~~~~~~~~~~~~~~~~~~~~\vdots \\
					&\mathcal{P}_{n}^{(n-2^{t})}=(p_{1}^{(n-2^{t})},\ldots,p_{n-2^{t}+1}^{(n-2^{t})}(=p_{1}),\ldots,p_{2^{t}}^{(n-2^{t})}).
				\end{aligned}\label{eq20}
			\end{equation}
		\end{small}\par
		Due to $n\le 2^{t+1}-1$, the number of probabilities for $n-2^t$ mergers, $2(n-2^t)$, is less than or equal to $n-1$. Thus, in $\mathcal{P}_{n}^{(n-2^{t})}$, $p_{1}$ has not been yet selected to merge. The code tree is as shown in Fig.~\ref{phase1}. By Theorem~\ref{theorem2}, $l_{1}$ equals codeword length of $p_{1}$ in $\mathcal{P}_{n}^{(n-2^{t})}$.		
\item Consider the last $2^{t}$ merges. Examine the code tree of $\mathcal{P}_{n}^{(n-2^{t})}$. We assert the tree is a perfect binary tree of depth $t$. Consider mathematical induction.\par
		Considering the characteristics of $\mathcal{P}_{n}^{(n-2^{t})}$, we obtain
		\begin{small}
			\begin{equation}
				\begin{aligned}
					p_{1}^{(n-2^{t})}&=p_{2^{t}}^{(n-2^{t}-1)}+p_{2^{t}+1}^{(n-2^{t}-1)}\\
					&\le p_{2^{t}-2}^{(n-2^{t}-1)}+p_{2^{t}-1}^{(n-2^{t}-1)}\\
					&=p_{2^{t}-1}^{(n-2^{t})}+p_{2^{t}}^{(n-2^{t})}\\
					&\le p_{2^{t}-3}^{(n-2^{t})}+p_{2^{t}-2}^{(n-2^{t})}.\\
				\end{aligned}\label{eq21}
			\end{equation}
		\end{small}Furthermore, for $0\le k\le 2^{t-1}-2$,
		\begin{small}
			\begin{equation}
				p_{2k+1}^{(n-2^{t})}+p_{2k+2}^{(n-2^{t})}\ge p_{2k+3}^{(n-2^{t})}+p_{2k+4}^{(n-2^{t})}\ge p_{1}^{(n-2^{t})}.\label{eq22}
			\end{equation}
		\end{small}\par	
		When $t=1$, $\mathcal{P}_{n}^{(n-2^{t})}$ contains only two probabilities.
		Merging these two nodes directly results in a full binary tree of depth $1$. The conclusion holds.
		
		Assume that for any $t-1$ (i.e., when the number of nodes is $2^{t-1}$), if the distribution satisfies \eqref{eq22}, its Huffman tree is a full binary tree of depth $t-1$.\par
		We now proceed to demonstrate that the same conclusion holds for $t$.
		Let $p'_{m}\triangleq p_{2m-1}^{(n-2^{t})}+p_{2m}^{(n-2^{t})}$, for $1\le m\le 2^{t-1}$. Obviously, for $0\le k\le 2^{t-2}-2$,
		\begin{equation}
			p'_{2k+1}+p'_{2k+2}\ge p'_{2k+3}+p'_{2k+4}\ge p'_{1}.\label{eq23}
		\end{equation}
		Consider the merging process of $\mathcal{P}_{n}^{(n-2^{t})}$, for $1\le i\le 2^{t-1}$,
		\begin{small}
			\begin{equation}
				\begin{aligned}
					&\mathcal{P}_{n}^{(n-2^{t}+1)}=\{p'_{2^{t-1}},p_{1}^{(n-2^{t})},\ldots,p_{2^{t}-2}^{(n-2^{t})}\}\\
					&\mathcal{P}_{n}^{(n-2^{t}+2)}=\{p'_{2^{t-1}-1},p'_{2^{t-1}},\ldots,p_{2^{t}-4}^{(n-2^{t})}\}\\
					&~~~~~~~~~~~~~~~~~~~~~~~~~~~~\vdots\\
					&\mathcal{P}_{n}^{(n-2^{t}+i)}=\{p'_{2^{t-1}-i+1},p'_{2^{t-1}-i+2},\ldots,p_{2^{t}-2i}^{(n-2^{t})}\}\\
					&~~~~~~~~~~~~~~~~~~~~~~~~~~~~\vdots\\
					&\mathcal{P}_{n}^{(n-2^{t-1})}=\{p'_{1},p'_{2},\ldots,p'_{2^{t-1}}\}\\
				\end{aligned}.\label{eq24}
			\end{equation}
		\end{small}\par
		Since \eqref{eq23} holds, we examine the probability distribution $\mathcal{P}_{n}^{(n-2^{t-1})}$ satisfying \eqref{eq22}, its Huffman tree is a full binary tree of depth $t-1$. The code tree is as shown in Fig.~\ref{phase2}.\par
		Therefore, the tree of $\mathcal{P}_{n}^{(n-2^{t})}$ is a perfect binary tree of depth $t$. From that, $l_{1}=t=\lfloor\log_{2}n\rfloor$.
		The proof is complete.
	\end{enumerate}
	\end{proof}
	Based on Lemma 2, combined with the concept of $\delta$-occasion and considering the specific value of $l_{1}$ in a general probability distribution, we obtain the following corollary.
	\begin{corollary}\label{corollary1}
		Consider a finite source $\mathcal{S}_{n}=(\mathcal{X}_{n},\mathcal{P}_{n})$, where $\mathcal{P}_{n}$ satisfies \eqref{eq12} and $p_{1}<\frac{1}{2}$. Then $l_{1}=\lfloor \log_{2}(n-\delta)\rfloor$.
	\end{corollary}
	\begin{proof}
		In the process of standardized Huffman coding of $\mathcal{P}_{n}$, considering its $\delta$-occasion, then from Definition~\ref{definition4}, the probability distributions satisfies \eqref{eq13} and \eqref{eq14}.
		Applying Lemma~\ref{lemma2} to $\mathcal{P}^{(\delta)}_{n}$, then the codeword length of $p^{(\delta)}_{1}$ is $l_{1,n-\delta}=\lfloor \log_{2}(n-\delta)\rfloor$.\par	
		Due to
			\begin{equation*}
				\begin{aligned}
					p_{1}&>p_{n-\delta+1}^{(\delta-1)}+p_{n-\delta}^{(\delta-1)}\\
					&\ge p_{n-\delta+2}^{(\delta-2)}+p_{n-\delta+1}^{(\delta-2)}\\
					& \quad \quad \quad  \quad \vdots\\
					&\ge p_{n}^{(0)}+p_{n-1}^{(0)}
				\end{aligned} ,
			\end{equation*}
so until $\mathcal{P}_{n}$ is merged into $\mathcal{P}^{(\delta)}_{n}$, $p_{1}$ has not yet been selected to merge.
	
By Theorem~\ref{theorem2}, $l_{1}$ equals codeword length $l_{1,n-i}$ in $\mathcal{P}_{n}^{(i)}$, for any $i\in[0,\delta]$. Then
		\begin{equation}
			l_{1}=l_{1,n}=l_{1,n-1}=\cdots=l_{1,n-\delta}=\lfloor \log_{2}(n-\delta)\rfloor.
		\end{equation}
		The proof is complete.
	\end{proof}
	
	\begin{lemma}\label{lemma3}
		Consider a finite source $\mathcal{S}_{n}=(\mathcal{X}_{n},\mathcal{P}_{n})$, where $\mathcal{P}_{n}$ satisfies \eqref{eq12} and $p_{1}<\frac{1}{2}$. If $p_{1}<p$, then $l_{1}\ge\lfloor-\log_{2}p\rfloor$.
	\end{lemma}
	\begin{proof}
	From Lemma~\ref{lemma1}, if $p_{1}<p$, then $\delta<n-\frac{1}{p}$.  Meanwhile, from Corollary~\ref{corollary1},
	\begin{equation}
		l_{1}=\lfloor\log_{2}(n-\delta)\rfloor\ge \lfloor-\log_{2}p\rfloor.\label{eq25}
	\end{equation}
	The proof is complete.
	\end{proof}
	\section{Main Results}
	We provide the main conclusions to be proven in this paper. First, we are the first to prove the existence of a set of intervals such that when the maximum probability $p_{1}$ falls in any of these intervals, its optimal codeword length can be immediately and uniquely determined.
	\subsection{Interval of the Maximum Probability}
	\begin{theorem}\label{theorem4}
		The finite source $\mathcal{S}_{n}=(\mathcal{X}_{n},\mathcal{P}_{n})$ satisfies \eqref{eq12}. 
If $p_{1}\in(\frac{2}{2^{k+1}+1},\frac{1}{2^{k}-1})$, then $l_{1}=k$.
	\end{theorem}
	\begin{proof}
	From Lemma~\ref{lemma1}, since $p_{1}\in(\frac{2}{2^{k+1}+1},\frac{1}{2^{k}-1})$, we deduce that\footnote{If $p_{1}\ge\frac{1}{2}$, the $\delta$-occasion does not actually exist. However, since $l_{1}$ is determined to be 1 in this case, which is consistent with the theorem, the following discussion will focus exclusively on the situation where $p_{1}<\frac{1}{2}$.}
		\begin{equation}
			\begin{aligned}
				\delta&\in(n-2^{k+1},n-2^{k}+1),\\
				&n-\delta\in(2^{k}-1,2^{k+1}).\\
			\end{aligned}\label{eq26}
		\end{equation}
		Due to $n-\delta\in \mathbb{N}$, we obtain
		\begin{equation}
			n-\delta\in[2^{k},2^{k+1}-1].\label{eq27}
		\end{equation}
		Consequently, by Corollary~\ref{corollary1}, we have
		\begin{equation}
			l_{1}\in[\lfloor\log_{2}2^{k}\rfloor,\lfloor\log_{2}(2^{k+1}-1)\rfloor];\label{eq28}
		\end{equation}
		that is $l_{1}=k$.
	\end{proof}
	\begin{remark}
		The sum of the lengths of all probability intervals of Theorem~\ref{theorem4} can be calculated.
		\begin{small}
			\begin{equation*}
				\begin{aligned}		
& \quad \sum_{k=1}^{\infty}(\frac{1}{2^{k}-1}-\frac{2}{2^{k+1}+1})  \\ &=(\sum_{k=1}^{10}\frac{1}{2^{k}-1}+\sum_{k=11}^{\infty}\frac{1}{2^{k}-1})-(\sum_{k=1}^{10}\frac{2}{2^{k+1}+1}+\sum_{k=11}^{\infty}\frac{2}{2^{k+1}+1}).\\
				\end{aligned}
			\end{equation*}
		\end{small}
Since\begin{small}
\begin{equation*}
	\begin{aligned}
		&\sum_{k=11}^{\infty}\frac{1}{2^{k}-1}<\sum_{k=11}^{\infty}\frac{1}{2^{k-1}}=\frac{1}{2^{9}},\\
&\sum_{k=11}^{\infty}\frac{2}{2^{k+1}+1}=\sum_{k=11}^{\infty}\frac{1}{2^{k}+\frac{1}{2}}<\sum_{k=11}^{\infty}\frac{1}{2^{k}}=\frac{1}{2^{10}},
	\end{aligned}
\end{equation*}
\end{small}therefore, we obtain
\begin{small}
	\begin{equation*}
	\begin{aligned}		
& \quad \sum_{k=1}^{\infty}(\frac{1}{2^{k}-1}-\frac{2}{2^{k+1}+1})  \\ 
&>\sum_{k=1}^{10}\frac{1}{2^{k}-1}-\sum_{k=1}^{10}\frac{2}{2^{k+1}+1}-\frac{1}{2^{10}}\\
&>0.743385,    \\
	\end{aligned}
	\end{equation*}
	\end{small}
and
\begin{small}
	\begin{equation*}
	\begin{aligned}		
& \quad \sum_{k=1}^{\infty}(\frac{1}{2^{k}-1}-\frac{2}{2^{k+1}+1})  \\
&<\sum_{k=1}^{10}\frac{1}{2^{k}-1}+\frac{1}{2^{9}}-\sum_{k=1}^{10}\frac{2}{2^{k+1}+1}\\
&<0.746315.    \\
	\end{aligned}
	\end{equation*}
\end{small}
 This implies that according to Theorem~\ref{theorem4}, the probability interval of $p_{1}$ that $l_{1}$ can be immediately determined occupies nearly three-quarters of $(0,1)$.
\end{remark}
	In fact, if $p_{1}$ does not fall in any of the above intervals, $l_{1}$ cannot be uniquely determined. For the case of $k=2$, counterexamples can be provided to show that if $p_{1}\notin(\frac{2}{9},\frac{1}{3})$, then $l_{1}=2$ is not necessarily guaranteed.\par
	\textit{Counterexample 1:} If $\frac{1}{3}\le p_{1}<\frac{1}{2}$
	then we construct a distribution $\mathcal{P}_{1}$ with $0\le \epsilon<\frac{1}{6}$.
	\begin{small}
		\begin{equation}
			p_{i}=\left\{
			\begin{aligned}
				&\frac{1}{3}+\epsilon,&& i=1\\
				&\frac{1}{3},&& i=2 \\
				&\frac{1}{3}-\epsilon,&& i=3
			\end{aligned}.
			\right.\label{eq29}
		\end{equation}
	\end{small}
	Clearly, $l_{1}=1$.\par
	\begin{figure}[!t]
		\centering
		\includegraphics[width=0.25\textwidth]{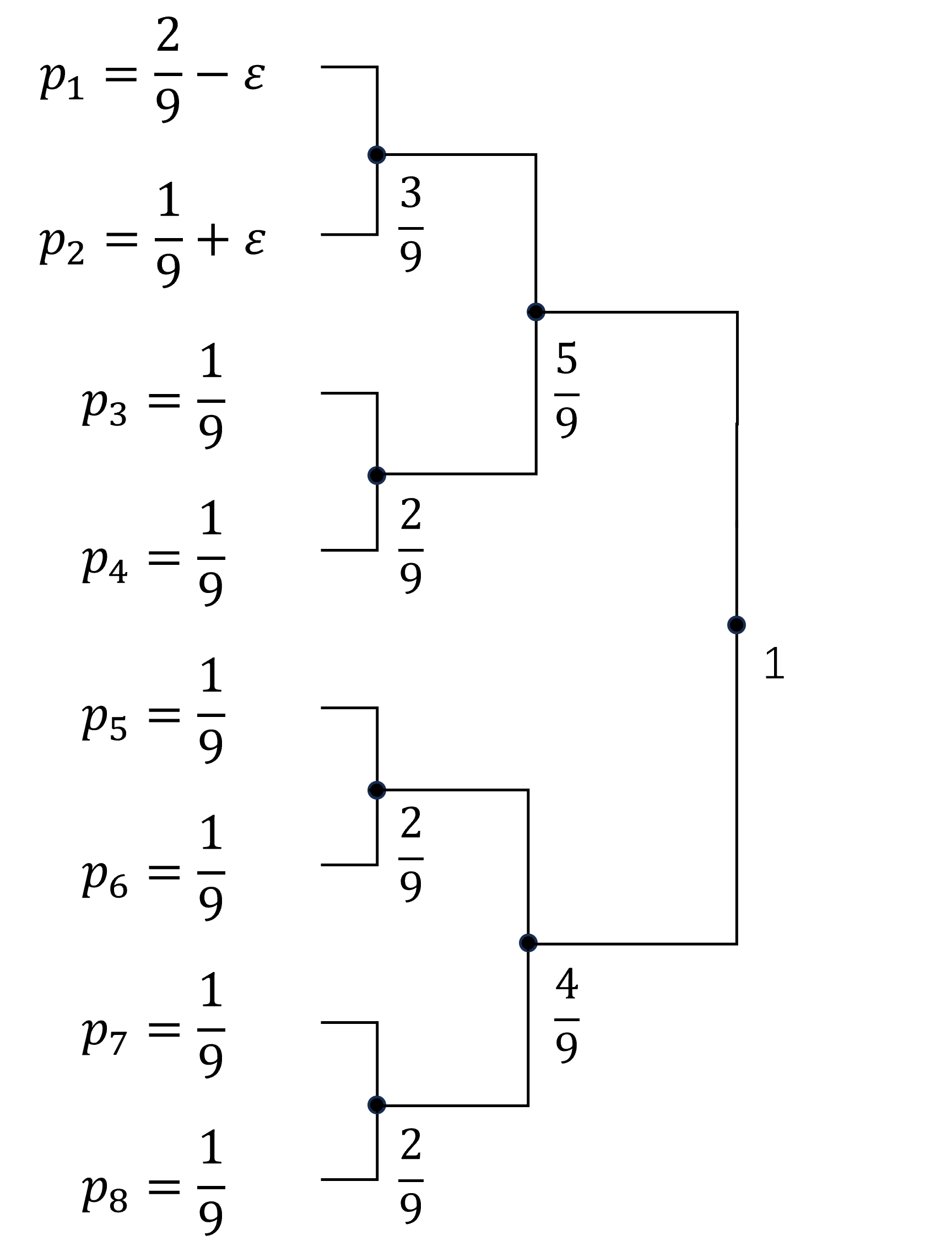}
		\caption{Counterexample 2 the corresponding coding tree is as shown in the figure, clearly, $l_{1}=3$.}
		\label{fig2}
	\end{figure}
	\begin{figure}[!t]
		\centering
		\includegraphics[width=0.35\textwidth]{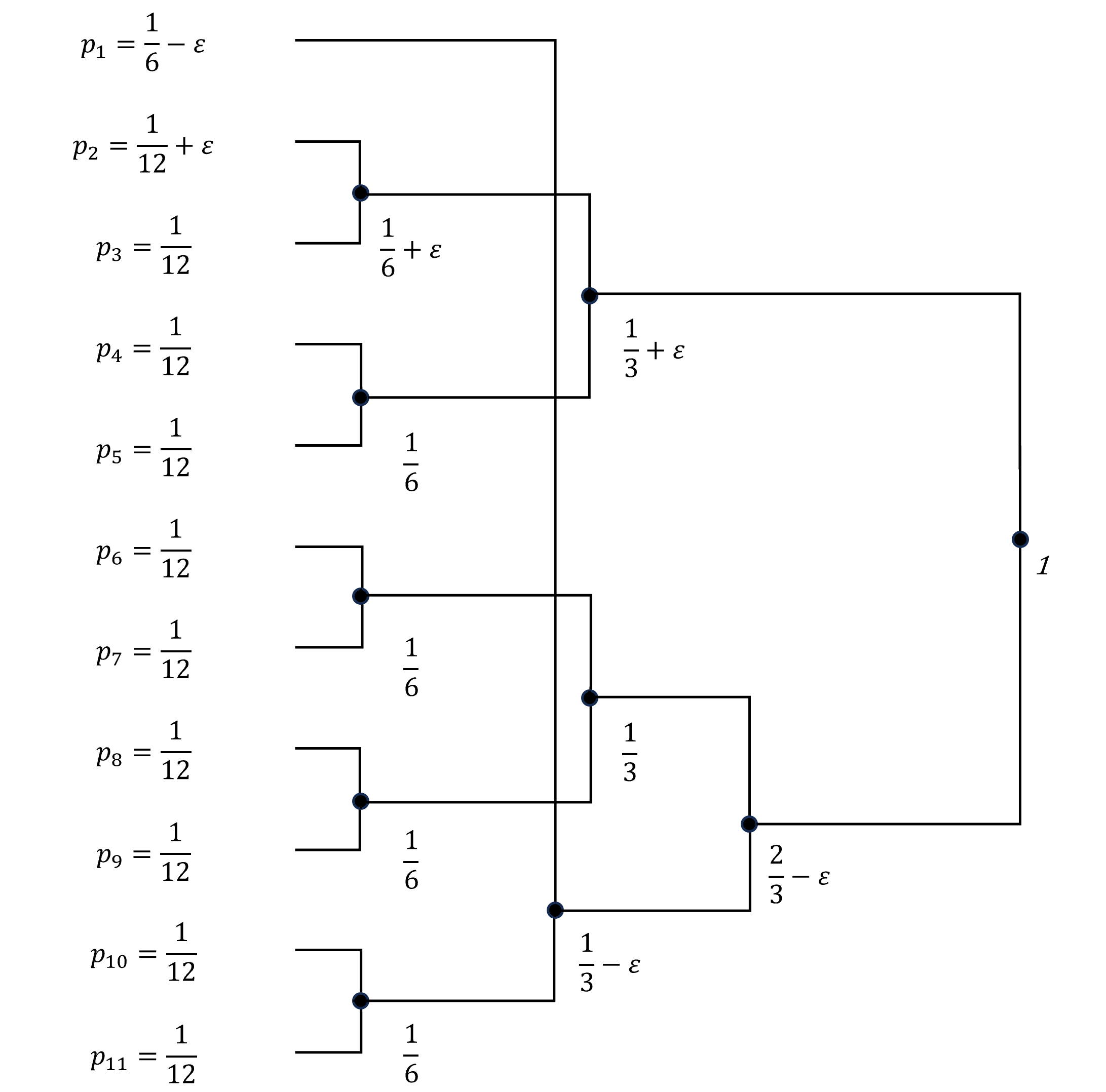}
		\caption{Counterexample 3 the corresponding coding tree is as shown in the figure, clearly, $l_{1}=3$.}
		\label{fig3}
	\end{figure}
	\textit{Counterexample 2:} If $\frac{1}{6}<p_{1}\le\frac{2}{9}$
	then we construct a distribution $\mathcal{P}_{2}$ with $0\le\epsilon<\frac{1}{18}$.
	\begin{small}
	\begin{equation}
		p_{i}=\left\{
		\begin{aligned}
			&\frac{2}{9}-\epsilon,&& i=1\\
			&\frac{1}{9}+\epsilon,&& i=2 \\
			&\frac{1}{9},&& 3\le i\le 8
		\end{aligned}.
		\right.\label{eq30}
	\end{equation}
	\end{small}
	The code tree is as shown in Fig.~\ref{fig2}. Clearly, $l_{1}=3$.\par
	\textit{Counterexample 3:} If $\frac{1}{8}\le p_{1}\le\frac{1}{6}$
	then we construct a distribution $\mathcal{P}_{3}$ with $0\le\epsilon\le\frac{1}{24}$.
	\begin{small}
		\begin{equation}
			p_{i}=\left\{
			\begin{aligned}
				&\frac{1}{6}-\epsilon,&& i=1\\
				&\frac{1}{12}+\epsilon,&& i=2 \\
				&\frac{1}{12},&& 3\le i\le 11
			\end{aligned}.
			\right.\label{eq31}
		\end{equation}
	\end{small}
	The code tree is as shown in Fig.~\ref{fig3}. Clearly, $l_{1}=3$.\par
	\textit{Counterexample 4:} If $p_{1}<\frac{1}{8}$, then by Lemma~\ref{lemma3}, we have $l_{1}\ge 3$.\par
	On the basis of Theorem~\ref{theorem4}, we can prove the following result in the case of infinite sources.
	\begin{theorem}\label{theorem5}
		Consider an infinite source $\mathcal{S}=(\mathcal{X},\mathcal{P})$. If
		\begin{small}
			\begin{equation}
				p_{1}\in(\frac{2}{2^{k+1}+1},\frac{1}{2^{k}-1}),\label{eq32}
			\end{equation}
		\end{small}then $l_{1}^{best}=k$.
	\end{theorem}
	\begin{proof}
	Let $\mathcal{S}_{n}=(\mathcal{X}_{n},\mathcal{P}_{n})$ denote the truncated version of $\mathcal{S}$ with size $n$, where $\mathcal{P}_{n}=(q_{1},q_{2},...,q_{n})$ satisfies \eqref{eq6} and $C_{n}$ be Huffman code for $\mathcal{X}_{n}$.\par
	By Theorem~\ref{theorem1}, $\{C_{n}\}_{n=1}^{\infty}$ contains a subsequence $C_{n_{m}}\rightarrow C_{best}$. Meanwhile, the sequence $\{S_{n_{m}}\}_{m=1}^{\infty}$ is increasing and satisfies
	\begin{small}
		\begin{equation}
			S_{\infty}=\displaystyle\sum_{j=1}^{\infty}p_{j}=1.\label{eq33}
		\end{equation}
	\end{small}Hence $S_{n_{m}}\rightarrow1$, we have $p_{1}/S_{n_{m}}\rightarrow p_{1}.\label{eq34}$
	Let
	\begin{small}
		\begin{equation}
			\epsilon\triangleq\min\left\{p_{1}-\frac{2}{2^{k+1}+1},\frac{1}{2^{k}-1}-p_{1}\right\},\label{eq35}
		\end{equation}
	\end{small}there exists $N$ such that for all $n_{m}>N,\,|\frac{p_{1}}{S_{n_{m}}}-p_{1}|<\epsilon$.
	Therefore, we obtain
	\begin{small}
		\begin{equation}
			\begin{aligned}
				p_{1}\!-\!p_{1}\!+\!\frac{2}{2^{k+1}\!+\!1}\leq p_{1}\!-\!\epsilon<\frac{p_{1}}{S_{n_{m}}}&<p_{1}\!+\!\epsilon \leq p_{1}\!+\!\frac{1}{2^{k}\!-\!1}\!-\!p_{1};\\
			\end{aligned}\label{eq36}
		\end{equation}
	\end{small}that is
	\begin{equation}
		\frac{p_{1}}{S_{n_{m}}}\in(\frac{2}{2^{k+1}+1},\frac{1}{2^{k}-1}).\label{eq37}
	\end{equation}
	By Theorem~\ref{theorem4}, for any $n_{m}>N$, $C_{n_{m}}$ satisfies $l_{1,n_{m}}=k$. From Definition~\ref{definition1}, $l_{1}^{best}=k$.
	\end{proof}
	\subsection{A Class of Recursively Structured Probability Distributions}
	Consider a class of infinite sources $\mathcal{S}=(\mathcal{X},\mathcal{P})$, where $\mathcal{P}$ satisfies \eqref{eq1}. Let $\alpha_{1}\triangleq p_{1},\,\alpha_{m}\triangleq\frac{p_{m}}{1-\sum_{j=1}^{m-1}p_{j}}(m\ge 2)$. Below, we use $\{\alpha_{i}\}_{i=1}^{\infty}$ to represent $(p_{1},p_{2},...)$.
	\begin{enumerate}
\item If $m=1,\,p_{1}=\alpha_{1}$. If $m=2,\,p_{2}=\alpha_{2}(1-\alpha_{1})$.
\item  Assume for $2\le k\le m-1,\, p_{k}=\alpha_{k}\displaystyle\prod_{j=1}^{k-1}(1-\alpha_{j})$, then
	\begin{small}
		\begin{equation}
			\begin{aligned}
				p_{m}&=\alpha_{m}(1-\sum_{j=1}^{m-1}p_{j})\\
				&=\alpha_{m}(1-\alpha_{1}-\sum_{j=2}^{m-1}p_{j})\\
				&=\alpha_{m}(1-\alpha_{1}-\sum_{j=2}^{m-1}\alpha_{j}\prod_{l=1}^{j-1}(1-\alpha_{l}))\\
				&=\alpha_{m}(1-\alpha_{1})(1-\alpha_{2}-\sum_{j=3}^{m-1}\alpha_{j}\prod_{l=2}^{j-1}(1-\alpha_{l}))\\
				&~~~~~~~~~~~~~~~~~~~~~~~~~\vdots\\
				&=\alpha_{m}\displaystyle\prod_{j=1}^{m-1}(1-\alpha_{j})
			\end{aligned}.\label{eq38}
		\end{equation}
	\end{small}
\end{enumerate}
	Therefore, we obtain\footnote{If for all $i,j$ we have $\alpha_{i}=\alpha_{j}$, then the probability distribution reduces to a geometric one.}
	\begin{small}
		\begin{equation}
			p_{1}=\alpha_{1},\,p_{m}=\alpha_{m}\displaystyle\prod_{j=1}^{m-1}(1-\alpha_{j}),\,\, m\ge 2.\label{eq39}
		\end{equation}
	\end{small}

	\begin{theorem}\label{theorem6}
		If all $\alpha_{i}\in(0.4,1)$, then the optimal code lengths corresponding to the infinite source $\mathcal{S}=(\mathcal{X},\mathcal{P})$ is
		 $\{l_{i}^{best}=i\}_{i=1}^{\infty}$, where $\mathcal{P}$ satisfies
		 \eqref{eq39}.
	\end{theorem}
	\begin{proof}Let $\mathcal{S}_{n}=(\mathcal{X}_{n},\mathcal{P}_{n})$ denote the truncated version of $\mathcal{S}$, where $\mathcal{X}_{n}=\{1,2,...,n\}$ and $\mathcal{P}_{n}$ satisfies \eqref{eq6}.
	From \eqref{eq6}, we have $p_{i}= q_{i}S_{n}$.
	
	By Theorem~\ref{theorem1}, $\{C_{n}\}_{n=1}^{\infty}$ contains a subsequence $C_{n_{m}}\rightarrow C_{best}$. Consider the corresponding sequence $\{\mathcal{P}_{n_{m}}\}_{m=1}^{\infty}$,
	we obtain

	\begin{small}
		\begin{equation*}
			\begin{aligned}
		& \quad (q_{i+2}+q_{i+3}\cdots+q_{n_{m}})\times S_{n_{m}}  \\
&= p_{i+2}+p_{i+3}\cdots+p_{n_{m}}\\
		& =\sum_{t=i+2}^{n_{m}}\big[\alpha_{t}\prod_{j=1}^{t-1}(1-\alpha_{j})\big]  \\
&=\sum_{t=1}^{n_{m}-i-1}\big[\alpha_{t+i+1}\prod_{j=1}^{t+i}(1-\alpha_{j})\big]\\
				&<\sum_{t=1}^{n_{m}-i-1}\big[\alpha_{t+i+1}\prod_{j=1}^{t+i}(1-\alpha_{j})\big]+\prod_{t=1}^{n_{m}}(1-\alpha_{t})\\
				&=\prod_{j=1}^{i+1}(1\!-\!\alpha_{j})\left[\prod_{t=i+2}^{n_{m}}\!(1\!-\!\alpha_{t})\!+\!\alpha_{n_m}\prod_{t=i+2}^{n_{m}-1}\!(1\!-\!\alpha_{t})\!+\!\cdots\!+\!\alpha_{i+2}\right]    \\
			&=\prod_{j=1}^{i+1}(1\!-\!\alpha_{j})\left[\prod_{t=i+2}^{n_{m}-1}\!(1\!-\!\alpha_{t})\!+\!\alpha_{n_{m}\!-\!1}\!\prod_{t=i+2}^{n_{m}-2}\!(1\!-\!\alpha_{t})\!+\!\cdots\!+\!\alpha_{i+2}\right]    \\
				&=\prod_{j=1}^{i+1}(1\!-\!\alpha_{j}) \overset{(a)}{<}\alpha_{i}\prod_{j=1}^{i-1}(1-\alpha_{j})\\
				&=p_{i}=q_{i}\times S_{n_{m}},
			\end{aligned}\label{eq43}
		\end{equation*}
	\end{small}
for all $1\le i\le n_{m}-3$, where $(a)$ is due to $(1-\alpha_{i})(1-\alpha_{i+1})<0.6\times0.6<0.4<\alpha_{i}$. Thus, we have $q_{i}>q_{i+2}+\cdots+q_{n_{m}}$ for all $1\le i\le n_{m}-3$.
	From Theorem~\ref{theorem3}, $\mathcal{P}_{n_{m}}$ is an anti-uniform source, then for any $C_{n_m},\, l_{k,n_{m}}=k$.
	From Definition~\ref{definition1} and $C_{n_{m}}\rightarrow C_{best}$ , we obtain $l_{k}^{best}=k$. \end{proof}
	\begin{remark}
	The only place where the range of values of $\alpha_{i}$ is needed in the proof is $(1-\alpha_{i})(1-\alpha_{i+1})\leq\alpha_{i}$. Therefore, the condition of Theorem~\ref{theorem6} can be relaxed to $\alpha_{i}\in[\frac{3-\sqrt{5}}{2},1)$.	
	\end{remark}
	
	Determining the optimal prefix code according to the above theorem requires less information than the condition (3) for an anti-uniform source. For instance, to verify whether the codeword length $l_{k}=k$, (3) requires information of $\{p_{i}\}_{i=1}^{k+1}$, whereas the present result requires only $\{p_{i}\}_{i=1}^{k}$.\par
	\section{Conclusion}	
	In this paper, we extend optimal prefix code theory to countably infinite alphabets. First, for any positive integer $k$, we establish precise probability intervals for the largest symbol probability $p_{1}$ that guarantee the corresponding codeword length $l^{best}_{1}=k$ in the optimal code. This result bridges the finite-alphabet Huffman construction and the infinite-alphabet case via a sequence convergence argument. Second, we propose a novel recursive structural criterion for distributions whose optimal code lengths satisfy $\{l_{i}=i\}_{i=1}^{\infty}$, which provides a more efficient method for verifying the anti-uniform source than the known sufficient condition.

\end{document}